\documentclass[letterpaper, 10 pt, conference]{ieeeconf}

\IEEEoverridecommandlockouts 
\title{Transient performance of MPC for tracking without terminal constraints}
\author{Nadine Ehmann, Matthias K\"ohler, and Frank Allg\"ower
\thanks{F. Allg\"ower is thankful that this work was funded by the Deutsche Forschungsgemeinschaft (DFG, German Research Foundation) under Germany’s Excellence Strategy -- EXC 2075 -- 390740016.}
\thanks{Nadine Ehmann, Matthias K\"ohler and Frank Allg\"ower are with the University of Stuttgart, Institute for Systems Theory and Automatic Control, Germany \{nadine.ehmann, matthias.koehler, frank.allgower\}@ist.uni-stuttgart.de}}

\usepackage{cite}
\usepackage{amsmath,amssymb,amsfonts}
\usepackage{algorithmic}
\usepackage{graphicx}
\usepackage{hyperref}
\hypersetup{hidelinks=true}
\usepackage{siunitx}
\usepackage{pgfplots}
\usepackage{mathtools}
\usepackage{subcaption}
\usepackage{float}

\newtheorem{theorem}{Theorem}
\newtheorem{assumption}{Assumption}

\newtheorem{lemma}{Lemma}

\newtheorem{corollary}{Corollary}
\newtheorem{proposition}{Proposition}

\DeclareMathOperator*{\argmin}{arg\,min}
\DeclareMathOperator*{\arginf}{arg\,inf}

\def\BibTeX{{\rm B\kern-.05em{\sc i\kern-.025em b}\kern-.08em
		T\kern-.1667em\lower.7ex\hbox{E}\kern-.125emX}}

\newcommand\copyrighttext{%
	\footnotesize 
	\textcopyright 2025 IEEE. Personal use of this material is permitted. Permission from IEEE must be obtained for all other uses, in any current or future 
	media, including reprinting/republishing this material for advertising or promotional purposes, creating new collective works, for resale or redistribution to servers or lists, or reuse of any copyrighted component of this work in other works. DOI: \href{https://www.doi.org/10.1109/LCSYS.2025.3585945}{10.1109/LCSYS.2025.3585945}, N. Ehmann, M. Köhler, and F. Allgöwer, "Transient Performance of MPC for Tracking Without Terminal Constraints," \emph{IEEE Control Syst. Lett.}, vol. 9, pp. 2049-2054, 2025.
}
	
\newcommand\copyrightnotice{%
	\begin{tikzpicture}[remember picture,overlay]
		\node[anchor=south,yshift=10pt] at (current page.south) {\fbox{\parbox{\dimexpr\textwidth-\fboxsep-\fboxrule\relax}{\copyrighttext}}};
	\end{tikzpicture}%
}

\begin{document}

\maketitle
\copyrightnotice
\vspace*{-0.85\baselineskip}
\thispagestyle{empty}

\begin{abstract}
	Model predictive control (MPC) for tracking is a recently introduced approach, which extends standard MPC formulations by incorporating an artificial reference as an additional optimization variable, in order to track external and potentially time-varying references. In this work, we analyze the performance of such an MPC for tracking scheme without a terminal cost and terminal constraints. We derive a transient performance estimate, i.e.\ a bound on the closed-loop performance over an arbitrary time interval, yielding insights on how to select the scheme's parameters for performance. Furthermore, we show that in the asymptotic case, where the prediction horizon and observed time interval tend to infinity, the closed-loop solution of MPC for tracking recovers the infinite horizon optimal solution.
\end{abstract}

\section{Introduction}

Model predictive control (MPC) is a powerful control method that can achieve stability and constraint satisfaction for nonlinear systems based on repeatedly solving a finite horizon optimization problem including performance criteria and only applying the first part of the optimal input, see~\cite{Gruene2017} for an overview. An important control goal of MPC is tracking of time-varying references. In this context, MPC is applied, e.g. in process control, autonomous driving and robotics.

\emph{MPC for tracking}, first introduced in~\cite{Limon2008}, is a variant of MPC designed for this control goal. In MPC for tracking, an artificial reference is included into the problem as an additional optimization variable. The major advantage is that recursive feasibility and stability are ensured independent of the externally provided reference. Moreover, if this reference is inadmissible, e.g. not reachable or not an equilibrium of the system, then the closed loop automatically converges to the best reachable reference instead. Furthermore, it provides a large region of attraction and allows for possibly shorter prediction horizons. Following the work in~\cite{Limon2008}, MPC for tracking has been extended e.g. to nonlinear systems and output tracking~\cite{Limon2018}, periodic trajectories~\cite{Koehler2020b}, economic stage costs and distributed systems, see e.g.~\cite{Krupa2024,Koehler2024} for a summary. One challenge in MPC for tracking, however, is that terminal ingredients are needed for every possible artificial reference, which can be done e.g. by using reference generic terminal ingredients~\cite{Koehler2020a}. To avoid such a potentially complex offline design, a further variant was introduced in~\cite{Soloperto2023}, where MPC for tracking without terminal cost and constraints is considered, but no insight on how to choose the parameters for good performance is provided.

Regarding performance of such MPC for tracking schemes there are only a few results available. Local optimality is shown for linear~\cite{Ferramosca2009,Ferramosca2011} and nonlinear~\cite{Limon2018} systems, when using specific offset cost functions. Transient and asymptotic performance results for schemes with terminal ingredients are presented in~\cite{Koehler2023}. Transient performance estimates bound the closed-loop performance, measured by the stage cost, over a finite-time window. So far, there are no performance bounds available for MPC for tracking without terminal ingredients.

In this work, we close this gap and provide a bound for the closed-loop performance over a transient time interval. We show that in the asymptotic case, as the prediction horizon and the considered time interval tend to infinity, the infinite horizon optimal solution is recovered. Our analysis also gives new insights into how to choose the parameters of the MPC for tracking scheme, namely the prediction horizon and a parameter representing an upper bound on the tracking cost, which is included as a constraint. In particular, we demonstrate that choosing this bound too small leads to poor performance despite large prediction horizons. However, we prove that choosing this parameter and the prediction horizon sufficiently large recovers good performance. Furthermore, the importance of scaling the offset cost with the prediction horizon is shown. The results are illustrated in simulation with the example of a continuous stirred-tank reactor.

\emph{Notation:}
We denote the natural numbers (including $0$) by $\mathbb{N}$ ($\mathbb{N}_0$) and all nonnegative real numbers by $\mathbb{R}_{\geq 0}$. We use $\mathbb{I}_{a:b}$ for the set of integers from $a$ to $b$ with $a\leq b$. With $\vert v \vert_{\mathcal{V}'} = \min_{v'\in\mathcal{V}'} \Vert v - v' \Vert_{2}$ we denote the Euclidean distance of a vector $v \in \mathcal{V}$ to the set $\mathcal{V}'$, where $\mathcal{V}' \subseteq \mathcal{V}$ is a closed subset and $\mathcal{V}$ a normed space. We use $\vert v \vert_{v'}$ if $\mathcal{V}' = \{v'\}$ for some $v' \in  \mathcal{V}$. The interior of a set $\mathcal{A}$ is $\mathrm{int}\, \mathcal{A}$ and define $\mathcal{B}_{c}(\tilde{x}) = \left\lbrace x \in \mathbb{R}^n \mid \Vert x - \tilde{x} \Vert_{2} \leq c \right\rbrace$. We use comparison functions, e.g. $\mathcal{K}_\infty$, $\mathcal{L}$, see \cite{Kellett2014}.

\section{Problem formulation}

We consider a nonlinear, time-invariant, discrete-time system $x(t+1)=f(x(t), u(t))$ with state $x(t)\in X\subset\mathbb{R}^n$, control input $u(t)\in U\subset\mathbb{R}^m$, time $t\in\mathbb{N}_0$, continuous system dynamics $f:\mathbb{R}^n\times\mathbb{R}^m\to\mathbb{R}^n$ and with $X$ bounded and $U$ compact. The system is subject to pointwise-in-time state and input constraints $(x(t), u(t))\in\mathbb{Z}\subseteq\mathbb{R}^{n\times m}$, where $\mathbb{Z}\subseteq X\times U$ is a compact constraint set. The solution, when starting at a given initial state $x\in X$ and applying the input sequence $u=(u(0), u(1), \dots)\in U^K$, is denoted by $x_u(k, x)$ with $k\in\mathbb{I}_{0:K}$. When the initial state $x$ is obvious, we use $x_u(k)=x_u(k, x)$. If $(x_u(k, x), u(k))\in\mathbb{Z}$, $k\in\mathbb{I}_{0:K-1}$ and $x_u(K, x)\in X$ for $u\in U^K$, then this $u$ is called admissible, while the set of all admissible input sequences is denoted by $\mathbb{U}^K(x)$. The set $\mathbb{U}^K_{\mathbb{X}}(x)$ with some closed set $\mathbb{X}$ additionally requires that $x_u(K, x)\in\mathbb{X}$ for $K\in\mathbb{N}$ or $\lim_{K\to\infty}\vert x_u(K, x)\vert_{\mathbb{X}}=0$ for $K=\infty$. The set of admissible references is given by $\mathbb{S}=\left\lbrace(x, u)\in\mathbb{Z}_r \,\vert\, x=f(x, u)\right\rbrace$ with a user-chosen closed set $\mathbb{Z}_r\subseteq\mathrm{int}\, \mathbb{Z}$. Define $\vert r\vert_{\hat{r}} = \sqrt{\vert x_r\vert_{x_{\hat{r}}}^2+\vert u_r\vert_{u_{\hat{r}}}^2}$ to compare different references $r=(x_r, u_r)$, $\hat{r}=(\hat{x}_r, \hat{u}_r)\in\mathbb{S}$.

The control goal is to stabilize the system and steer it towards a given external reference $r_\mathrm{e}=(x_\mathrm{e}, u_\mathrm{e})$ while minimizing a performance objective given by the stage cost $\ell:X\times U\times\mathbb{S}\to\mathbb{R}_{\geq 0}$. Also, the constraints should be satisfied. If $r_\mathrm{e}\notin\mathbb{S}$, i.e.\ $r_\mathrm{e}$ is not an admissible equilibrium of the system or $r_\mathrm{e}$ is not reachable, the control goal is to steer the system as close as admissible towards the external reference $r_\mathrm{e}$. With an offset cost function $T:\mathbb{S}\to\mathbb{R}_{\geq 0}$, which has its minimum at $r_\mathrm{e}$, the goal is formulated as follows: minimize $T$, i.e.\ steer the system to the best reachable reference $r_\mathrm{d}=(x_\mathrm{d}, u_\mathrm{d})=\argmin_{r\in\mathbb{S}}T(r)$. As in~\cite{Koehler2023}, we assume that $r_\mathrm{d}$ is unique and that the offset cost is an adequate measure for the distance from the considered reference to the best reachable reference $\vert r \vert_{r_{\mathrm{d}}}$:
\begin{assumption}\label{asm:offset_cost_indication}
	There exist $\alpha_{\mathrm{lo}}^T,\,\alpha_{\mathrm{up}}^T \in \mathcal{K}_{\infty}$ such that for all $r\in\mathbb{S}$ (not necessarily for all equilibria, e.g.\ $r_{\mathrm{e}}$)
	\begin{equation}\label{eq:offset_cost_indication}
		\alpha_{\mathrm{lo}}^T(\vert r \vert_{r_{\mathrm{d}}}) \le T(r) \le \alpha_{\mathrm{up}}^T(\vert r \vert_{r_{\mathrm{d}}}).
	\end{equation}
\end{assumption}
\vspace*{\belowdisplayskip}

This assumption also implies $T(r_\mathrm{d})=0$. In our setting, \eqref{eq:offset_cost_indication}~is satisfied if e.g. $T(r) \hspace{-1pt}=\hspace{-1pt} \Vert r - r_{\mathrm{e}} \Vert^2_S - \bar{T}$ with $S\hspace{-1pt}\succ\hspace{-1pt} 0$ and $\bar{T} = \Vert r_{\mathrm{d}} - r_{\mathrm{e}} \Vert^2_S$ is chosen, because then $T(r_\mathrm{d})=0$ and $T(r)>0$ for all $r\neq r_\mathrm{d}\in \mathbb{S}$. Note that it is not required to know $r_\mathrm{d}$ and $\bar{T}$, and that assuming $T(r_\mathrm{d})=0$ is no restriction, since the MPC optimization problem remains unchanged when adding constants, compare~\cite{Koehler2023}.

Furthermore, the stage cost defines a performance measure, hence, the scheme should exhibit good performance. This is measured by
\begin{equation}\label{eq:performance_measure}
	J_K^\mathrm{d}(x, u)=\textstyle\sum_{k=0}^{K-1}\ell(x_u(k, x), u(k), r_\mathrm{d})
\end{equation} for input sequences $u$, where we are interested in estimating and comparing the closed-loop performance, i.e.~\eqref{eq:performance_measure} for the closed-loop input of MPC for tracking.

\section{MPC for tracking}

In this section, we present the MPC for tracking problem and show that the resulting closed loop is exponentially stable, ensuring convergence to the best reachable steady state.

The tracking stage cost $\ell(x, u, r)$ should fulfill the following assumptions, cf.~\cite[Assm.~2-4]{Koehler2023}:
\begin{assumption}\label{asm:stage_cost_lower_and_upper_bound}
	There exist $c_1^{\ell},\,c_2^{\ell}>0$ such that for all $(x, u)\in\mathbb{Z}$ and $r=(x_r, u_r)\in\mathbb{S}$
	\begin{equation}\label{eq:stage_cost_lower_and_upper_bound}
		c_1^{\ell} \vert x \vert_{x_r}^2 \le \ell^* (x, r) \le c_2^{\ell} \vert x \vert_{x_r}^2
	\end{equation}
	with $\ell^* (x, r) = \min_{u \in U\,\text{s.t.}\,(x,u)\in\mathbb{Z}}\ell(x, u, r)$.
\end{assumption}
\begin{assumption}\label{asm:stage_cost_difference_bound}
	There exist $c_3^{\ell},\,c_4^{\ell}>0$ such that for any $r_{1},\,r_{2} \in\mathbb{S}$ and $(x, u) \in \mathbb{Z}$
	\begin{equation}\label{eq:stage_cost_difference_bound}
		\ell(x, u, r_{1}) \le c_3^{\ell} \ell(x, u, r_{2})+ c_4^{\ell} \vert r_{1} \vert_{r_{2}}^2.
	\end{equation}
\end{assumption}
\vspace*{\belowdisplayskip}
\begin{assumption}\label{asm:stage_cost_difference_bound_with_linear_term}
	There exist $c_5^{\ell},\,c_6^{\ell}>0$ such that for any $r_{1},\,r_{2} \in\mathbb{S}$ and $(x, u) \in \mathbb{Z}$
	\begin{equation}\label{eq:stage_cost_difference_bound_with_linear_term}
		\ell(x, u, r_{1}) \le \ell(x, u, r_{2}) + c_5^{\ell} \vert r_{1} \vert_{r_{2}}^2 + c_6^{\ell} \vert r_{1} \vert_{r_{2}}.
	\end{equation}
\end{assumption}
\vspace*{\belowdisplayskip}

Assumption~\ref{asm:stage_cost_lower_and_upper_bound} is standard in stabilizing MPC, see e.g.\ \cite[Assm.~3.2]{Gruene2012} or \cite[Assm.~2]{Koehler2024}, and states that the stage cost $\ell$ should reflect the distance between the considered state and (artificial) reference state. Stage costs with respect to different references can be compared with Assumptions~\ref{asm:stage_cost_difference_bound} and \ref{asm:stage_cost_difference_bound_with_linear_term}. The latter gives a direct relation without the additional factor $c_3^\ell$ and includes the linear term $c_6^\ell\vert r_{1}\vert_{r_2}$. These assumptions are satisfied, e.g.\ for the (standard) quadratic stage cost $\ell(x, u, r)=(x-x_r)^\top Q(x-x_r)+(u-u_r)^\top R(u-u_r)$ with $Q\hspace{-1pt}\succ\hspace{-1pt}0$, $R\hspace{-1pt}\succeq\hspace{-1pt}0$ on bounded constraint sets, cf.~\cite{Soloperto2023},~\cite{Koehler2023}.

For comparison and notation, we also briefly introduce the standard MPC problem without terminal ingredients with respect to a fixed reference $r\in\mathbb{S}$. This is given by \begin{subequations} \label{eq:standard_mpc_problem}
	\begin{align}
		&\mathcal{J}_N^\mathrm{s}(x, r) = \min_{u} J_N(x, u, r)\\
		&\text{s.t.}\,(x_u(k, x), u(t))\in\mathbb{Z},\;\,k\in\mathbb{I}_{0:N-1}
	\end{align}
\end{subequations} with a prediction horizon $N\in\mathbb{N}_0$ and the tracking cost \begin{equation}\label{eq:tracking_cost}
	J_N(x, u, r)=\textstyle\sum_{k=0}^{N-1} \ell( x_u(k, x), u(k), r).
\end{equation} The solution of~\eqref{eq:standard_mpc_problem} is the optimal input sequence $u_N^\mathrm{s}(\cdot\vert t)$, where we abbreviate $u_N^\mathrm{s}(\cdot\vert t)=u_N^\mathrm{s}(\cdot\vert x(t))$. The set of all states such that~\eqref{eq:standard_mpc_problem} is feasible is $\mathcal{X}_N^\mathrm{s}$.

We assume local exponential cost controllability of the system as follows. \begin{assumption}\label{asm:local_exponential_cost_controllability}
	There exist $\sigma>0$, $\gamma\geq 1$, such that for any $r\in\mathbb{S}$, $N\in\mathbb{N}$, and any $x\in X$ satisfying $\ell^*(x, r)\leq\sigma$, \begin{equation}\label{eq:local_exponential_cost_controllability}
		\mathcal{J}_N^\mathrm{s}(x, r)\leq\gamma\ell^*(x, r).
	\end{equation}
\end{assumption} 
\vspace*{\belowdisplayskip}

This assumption is standard for MPC without terminal cost and constraints, cf.~\cite[Assm.~2]{Soloperto2023}, \cite[Assm.~14]{Koehler2024}, \cite[Assm.~1]{Boccia2014} or \cite[Assm.~3.5]{Gruene2012}. In addition, it holds, e.g. if terminal cost and constraints could be calculated~\cite{Soloperto2023}.  If also $\mathcal{J}_N^\mathrm{s}(x, r)\leq\eta$ in addition to Assumption~\ref{asm:local_exponential_cost_controllability}, we can combine these conditions to $\mathcal{J}_N^\mathrm{s}(x, r)\leq \max\lbrace\gamma, \frac{\eta}{\sigma}\rbrace\ell^*(x, r)$ for all $x$ with $\mathcal{J}_N^\mathrm{s}(x, r)\leq\eta$, see~\cite{Soloperto2023,Boccia2014}. We choose $\eta\geq\gamma\sigma$, implying $\mathcal{J}_N^\mathrm{s}(x, r)\leq \frac{\eta}{\sigma}\ell^*(x, r)$.

To show stability, we use that if $r\neq r_\mathrm{d}$, there~always is another reference $\hat{r}$ that is close to the current artificial reference but has a smaller offset cost. Then, the artificial reference, that the closed loop is tracking, can be moved towards the best reachable reference. The following assumption, cf.~\cite[Assm.~7]{Koehler2023}, \cite[Assm.~3]{Soloperto2023}, ensures this.
\begin{assumption}\label{asm:better_candidate_reference}
	There exist $c_{1}^\mathrm{r},\,c_{2}^\mathrm{r} > 0$ such that for any $r=(x_r, u_r)\in\mathbb{S}$, and any $\theta \in [0, 1]$, there exists $\hat{r}=(x_{\hat{r}}, u_{\hat{r}})\in\mathbb{S}$ with
	\vspace{-3pt}\begin{align}
		\vert\hat{r}\vert_{r} &\le c_{1}^\mathrm{r}\theta\vert r\vert_{r_{{\mathrm{d}}}},\label{eq:candidate_reference_upper_bound}\\
		T(\hat{r})-T(r) &\le -c_{2}^\mathrm{r}\theta\vert r\vert_{r_{{\mathrm{d}}}}^2.\label{eq:candidate_reference_decrease}
	\end{align}
\end{assumption}
\vspace*{\belowdisplayskip}

Assumption~\ref{asm:better_candidate_reference} is satisfied, e.g.\ for a (strongly) convex offset cost, a convex set of references, and a uniqueness condition, which is often assumed in MPC for tracking, see e.g.~\cite[Assm.~1-2]{Limon2018} and \cite[Assm.~6]{Koehler2020b}.

In the MPC for tracking scheme the following optimization problem is solved at each time $t$ for a state $x$
\vspace{-1pt}\begin{subequations}\label{eq:MPC_problem}
	\begin{align}
		&H_{N,\eta}^*(x)=\min_{u, r} J_N(x, u, r) +\lambda(N)T(r) \label{eq:MPC_problem_optcost}\\[-2pt]
		&\mathrm{s.t.}\;r\in\mathbb{S},\\[-2pt]
		&\phantom{\mathrm{s.t.}\;}J_N(x, u, r) \leq \eta,\label{eq:MPC_problem_costconstraint}\\[-2pt]
		&\phantom{\mathrm{s.t.}\;}(x_{u}(k, x), u(k)) \in\mathbb{Z},\;\,k\in\mathbb{I}_{0:N-1}
	\end{align}
\end{subequations} with the tracking cost~\eqref{eq:tracking_cost} and prediction horizon $N\in\mathbb{N}_0$. The solution of~\eqref{eq:MPC_problem} is the optimal input trajectory $u_{N,\eta}^*(\cdot\vert t)$ together with the optimal artificial reference $r_{N,\eta}^*(x)=( x_{r_{N,\eta}^*(x)}, u_{r_{N,\eta}^*(x)})$. $\mathcal{X}_{N,\eta}$ denotes the set of states for which problem \eqref{eq:MPC_problem} is feasible. This MPC for tracking problem is similar to the one in~\cite{Soloperto2023}, where~\eqref{eq:MPC_problem_costconstraint} with the cost parameter $\eta$ ensures the evolution of the state inside the region of attraction of the chosen $r$. Together with $N$ it will also be relevant for performance. The key difference is that the offset cost is scaled by a function $\lambda:\mathbb{N}_0 \to \mathbb{R}_{\ge 0}$, depending on the prediction horizon. \begin{assumption}\label{asm:scaling}
	The scaling function $\lambda$ in~\eqref{eq:MPC_problem_optcost} satisfies $\lambda(0) \geq 1$ and $\lambda(N)\to\infty$ for $N\to\infty$.
\end{assumption} 

As shown below, it is essential for performance that the scaling grows unbounded (cf.~Section~\ref{sec:example}).

Note that if $r_{N,\eta}^*(x)$ is fixed, $\mathcal{J}_N^\mathrm{s}(x, r_{N,\eta}^*(x))=J_N(x, u_{N,\eta}^*(\cdot\vert t), r_{N,\eta}^*(x))$ and $u_N^\mathrm{s}(\cdot\vert t)=u_{N,\eta}^*(\cdot\vert t)$, and similar for candidate artificial references.

The closed-loop system \begin{equation*}
	x(t+1)=f(x(t), \mu_{N,\eta}( x(t))) 
\end{equation*} results by applying $\mu_{N,\eta}(x(t))=u_{N,\eta}^*(0\vert t)$ after solving~\eqref{eq:MPC_problem} for $x=x(t)$ at time step $t$. The following theorem ensures exponential stability. In this context, the parameter $\eta$ in~\eqref{eq:MPC_problem_costconstraint} influences the size of the feasible set $\mathcal{X}_{N,\eta}$, but it also affects the required prediction horizon length $N>N_\eta=\max\lbrace\frac{\eta}{\sigma}, \frac{\eta}{\sigma}(\gamma-1)\rbrace$.

\begin{theorem}\label{thm:exponential_stability}
	Let Assumptions \ref{asm:offset_cost_indication}\,--\,\ref{asm:stage_cost_difference_bound} and \ref{asm:local_exponential_cost_controllability}\,--\,\ref{asm:scaling} hold, consider $\eta\geq\gamma \sigma$ and $N>N_\eta$, and suppose that $x\in\mathcal{X}_{N,\eta}$, i.e.\ problem \eqref{eq:MPC_problem} is feasible at time $t=0$. Then, the MPC for tracking scheme is recursively feasible and the best reachable steady state $x_\mathrm{d}$ is exponentially stable with region of attraction $\mathcal{X}_{N,\eta}$. So there exist $c_\mathrm{exp}>0$ and $\gamma_\mathrm{exp}\in(0, 1)$ such that \begin{equation*}
		\vert x_{\mu_{N,\eta}}(t, x)\vert_{x_\mathrm{d}}\leq c_\mathrm{exp}\vert x\vert_{x_\mathrm{d}}\gamma_\mathrm{exp}^t
	\end{equation*} for all $x\in\mathcal{X}_{N,\eta}$, $t\in\mathbb{N}_0$. Also, $\lim_{t\to\infty} H_{N,\eta}^*(x(t))=0$.
\end{theorem}
\begin{proof}
	The proof follows~\cite[Thm.~2]{Soloperto2023} with $\bar{H}_{N,\eta}^*(x)=H_{N,\eta}^*(x)$ due to Assumption~\ref{asm:offset_cost_indication}, replacing $T_{\mathbb{Y}^\mathrm{d}}(y_t^\mathrm{s})$ with $\lambda(N)T(r)$ and when using $\lambda(N)\geq 1$ from Assumption~\ref{asm:scaling} as in~\cite[Lem.~2]{Koehler2023}.
\end{proof}

\begin{corollary}
	Under the assumptions of Theorem~\ref{thm:exponential_stability}, there exists $p\hspace{-1pt}\in\hspace{-1pt}\mathbb{I}_{1:N\hspace{-1pt}-\hspace{-1pt}1}$ such that for $x$ and $r$ with $\mathcal{J}_N^\mathrm{s}(x, r)\hspace{-1pt}\leq\hspace{-1pt}\eta$, 	\begin{equation}\label{eq:ell_star_lower_sigma}
		\ell^*\big(x_{u_N^\mathrm{s}(\cdot\vert t)}(p, x(t)), r\big)<\sigma
	\end{equation} with $\sigma$ from Assumption~\ref{asm:local_exponential_cost_controllability}. Also, for $k\in\mathbb{N}_0$ \begin{align}\label{eq:decrease_H}
		&\ell(x_{\mu_{N,\eta}}(k), \mu_{N,\eta}(x_{\mu_{N,\eta}}(k)), r_{N,\eta}^*(x_{\mu_{N,\eta}}(k)))\notag\\
		&\leq\frac{1}{\alpha_N}(H_{N,\eta}^*(x_{\mu_{N,\eta}}(k))-H_{N,\eta}^*(x_{\mu_{N,\eta}}(k+1)))
	\end{align} with $\alpha_N=1-(\frac{\eta}{\sigma}(\gamma-1))/N$ and $x_{\mu_{N,\eta}}(k)=x_{\mu_{N,\eta}}(k, x)$. This is a standard result in MPC without terminal constraints, cf.~\cite{Gruene2008}, \cite[Chapter~6]{Gruene2017}.
\end{corollary}

The parameters $\sigma$ and $\gamma$, and therefore $\eta$ and $N$, depend on the stage cost and the system and can be calculated for specific examples, see e.g.~\cite{Worthmann2016}.

\section{Transient performance estimate}

This section contains our main results. We derive an estimate for the closed-loop performance of MPC for tracking over an arbitrary finite time interval $K$, which is measured by the cost functional~\eqref{eq:performance_measure} for $u=\mu_{N,\eta}$. This so called \textit{transient performance} refers to the performance of the closed-loop solution during the transient phase $k\in\mathbb{I}_{0:K-1}$ until it reaches the neighbourhood around the best reachable steady state $\vert x_{\mu_{N,\eta}}(k, x)\vert_{x_\mathrm{d}}\leq c_\mathrm{exp}\vert x\vert_{x_\mathrm{d}}\gamma_\mathrm{exp}^K$.

First, we note that the standard MPC scheme, see~\eqref{eq:standard_mpc_problem}, with respect to some general, fixed reference $r\in\mathbb{S}$ results in an exponentially stable closed loop, e.g.\ for $r=r_\mathrm{d}$.

\begin{proposition}\label{prop:exponential_stability_standard_mpc}
	Let Assumptions \ref{asm:stage_cost_lower_and_upper_bound} and \ref{asm:local_exponential_cost_controllability} hold. For any $\eta\geq\gamma\sigma$, any $N>N_\eta$ and any $x$ satisfying $\mathcal{J}_N^\mathrm{s}(x, r)\leq\eta$, the standard MPC scheme is recursively feasible and $x_r$ is exponentially stable, i.e.\ there exist $c_\mathrm{s}>0$ and $\gamma_\mathrm{s}\in(0, 1)$ such that for all $t\in\mathbb{N}_0$ \begin{equation}\label{eq:exponential_stability_standard_mpc}
		\vert x_{u_N^\mathrm{s}(0\vert t)}(t, x)\vert_{x_r}\leq c_\mathrm{s}\vert x\vert_{x_r}\gamma_\mathrm{s}^t.
	\end{equation}
\end{proposition}
\vspace*{\belowdisplayskip}
\begin{proof}
	The proof is standard, see e.g.~\cite{Boccia2014,Koehler2024}. Since $c_1^\ell\vert x\vert_{x_r}^2\leq \mathcal{J}_N^\mathrm{s}(x, r)\leq\frac{\eta}{\sigma} c_2^\ell\vert x\vert_{x_r}^2$ and $\mathcal{J}_N^\mathrm{s}(x(t+1), r)-\mathcal{J}_N^\mathrm{s}(x(t), r)\leq-\alpha_N c_1^\ell\vert x(t)\vert_{x_r}^2$, we can choose e.g. $c_\mathrm{s}=\sqrt{\frac{\eta}{\sigma} c_2^\ell/c_1^\ell}$, $\gamma_\mathrm{s}\hspace{-1pt}=\hspace{-1pt}\sqrt{1\hspace{-1pt}-\hspace{-1pt}\alpha_N c_1^\ell/(\frac{\eta}{\sigma} c_2^\ell)}$ independent of $r$.
\end{proof}

Next, we consider performance results for the standard MPC problem~\eqref{eq:standard_mpc_problem}. This is similar to~\cite[Thm.~8.39]{Gruene2017}, but we adapt it to stage costs satisfying Assumption~\ref{asm:stage_cost_lower_and_upper_bound} and focus on a bound for the value function $\mathcal{J}_N^\mathrm{s}(x, r_\mathrm{d})$.

\begin{proposition}\label{prop:performance_bound_standard_mpc}
	Let Assumptions~\ref{asm:stage_cost_lower_and_upper_bound} and \ref{asm:local_exponential_cost_controllability} hold. Then, for any $\eta>0$, $N>N_\eta$ there exists $\delta\in\mathcal{L}$ such that for all $x$ satisfying $\mathcal{J}_N^\mathrm{s}(x, r_\mathrm{d})\leq \eta$ and $K\in\mathbb{N}$ \begin{equation}\label{eq:performance_bound_standard_mpc}
		\mathcal{J}_N^\mathrm{s}(x, r_\mathrm{d})\leq \inf_{u\in\mathbb{U}_{\mathcal{B}_\kappa(x_\mathrm{d})}^K(x)} J_K^\mathrm{d}(x, u)+\delta(K)
	\end{equation} with $\kappa=c_\mathrm{s}\vert x\vert_{x_\mathrm{d}}\gamma_\mathrm{s}^K$ and $c_\mathrm{s}$, $\gamma_\mathrm{s}$ from Proposition~\ref{prop:exponential_stability_standard_mpc}.
\end{proposition}

The proof of Proposition~\ref{prop:performance_bound_standard_mpc} is shown in Appendix~\ref{appendix:proof_standard_performance}.

The next lemma establishes a relation between the value functions of standard MPC and MPC for tracking. It enables us to compare the performance for varying prediction horizons $N$ in the following, while \eqref{eq:MPC_problem_costconstraint} remains satisfied throughout.

\begin{lemma}\label{lem:comparison_value_functions}
	Let Assumptions~\ref{asm:offset_cost_indication}\,--\,\ref{asm:stage_cost_difference_bound} and \ref{asm:local_exponential_cost_controllability}\,--\,\ref{asm:scaling} hold. Then, for all $\tilde{\eta}\geq\gamma\sigma$ and $\tilde{N}>N_{\tilde{\eta}}$ with $\mathcal{X}_{\tilde{N},\tilde{\eta}}\neq\emptyset$, there exists $\hat{\eta}$ such that the standard MPC problem \eqref{eq:standard_mpc_problem} with respect to $r_\mathrm{d}$ is feasible and for all $x\in\mathcal{X}_{\tilde{N},\tilde{\eta}}$, $\eta\geq\hat{\eta}$ and $N\in\mathbb{N}$ \begin{equation}
		H_{N,\eta}^*(x)\leq\mathcal{J}_N^\mathrm{s}(x, r_\mathrm{d})\leq \hat{\eta}.
	\end{equation}
\end{lemma}
\vspace*{\belowdisplayskip}

The proof of Lemma~\ref{lem:comparison_value_functions} is shown in Appendix~\ref{appendix:proof_comparison_value_functions}.

Lemma~\ref{lem:comparison_value_functions} shows that the value function of the MPC for tracking problem~\eqref{eq:MPC_problem} can be upper bounded by the value function of the standard MPC problem~\eqref{eq:standard_mpc_problem}. Furthermore, it can be bounded by $\hat{\eta}$, where this depends only on a fixed initial feasible set $\mathcal{X}_{\tilde{N},\tilde{\eta}}$, defined by $\tilde{N}$ and $\tilde{\eta}$, but not on $N$. This also indicates that the parameter $\eta$ in~\eqref{eq:MPC_problem_costconstraint} has to be chosen large enough.

These results are now combined to establish the transient performance bound for the closed loop:

\begin{theorem}\label{thm:transient_performance_bound}
	Suppose Assumptions~\ref{asm:offset_cost_indication}\,--\,\ref{asm:scaling} hold. Then, for any $\tilde{\eta}\geq\gamma\sigma$ and $\tilde{N}>N_{\tilde{\eta}}$, there exist $\hat{\eta}$ and $\delta\in\mathcal{L}$, such that for all $\eta\geq\hat{\eta}$, $N>N_\eta$, $x\in\mathcal{X}_{\tilde{N},\tilde{\eta}}$ and $K\in\mathbb{N}_0$ \begin{align}\label{eq:transient_performance_bound}
		&J_K^\mathrm{d}(x, \mu_{N,\eta})\leq \frac{1}{\alpha_N}\big(\inf_{u\in\mathbb{U}_{\mathcal{B}_\kappa(x_\mathrm{d})}^K(x)}J_K^\mathrm{d}(x, u)+\delta(K)\big)\notag\\
		&+\hspace{-1pt}\sum_{k=0}^{K-1}\hspace{-1pt}(c_5^\ell\hspace{-1pt}\left\vert r_{N,\eta}^*\hspace{-1pt}(x_{\mu_{N,\eta}}(k))\right\vert_{r_\mathrm{d}}^2\hspace{-3pt}+\hspace{-1pt}c_6^\ell\hspace{-1pt}\left\vert r_{N,\eta}^*\hspace{-1pt}(x_{\mu_{N,\eta}}(k))\right\vert_{r_\mathrm{d}})
	\end{align} with $\kappa=c_\mathrm{s}\vert x\vert_{x_\mathrm{d}}\gamma_\mathrm{s}^K$ and $c_\mathrm{s}$, $\gamma_\mathrm{s}$ from Proposition~\ref{prop:exponential_stability_standard_mpc}.
\end{theorem}
\begin{proof}
	With \eqref{eq:decrease_H} and Assumption~\ref{asm:stage_cost_difference_bound_with_linear_term} follows\\ 
	$J_K^\mathrm{d}(x, \mu_{N,\eta})= \sum_{k=0}^{K-1}\ell(x_{\mu_{N,\eta}}(k), \mu_{N,\eta}(x_{\mu_{N,\eta}}(k)), r_\mathrm{d})\\
	\stackrel{\mathclap{\eqref{eq:stage_cost_difference_bound_with_linear_term}}}{\leq} \sum_{k=0}^{K-1}\ell(x_{\mu_{N,\eta}}(k), \mu_{N,\eta}(x_{\mu_{N,\eta}}(k)), r_{N,\eta}^*(x_{\mu_{N,\eta}}(k)))+\\
	\sum_{k=0}^{K-1}(c_5^\ell\left\vert r_{N,\eta}^*(x_{\mu_{N,\eta}}(k))\right\vert_{r_\mathrm{d}}^2+c_6^\ell\left\vert r_{N,\eta}^*(x_{\mu_{N,\eta}}(k))\right\vert_{r_\mathrm{d}})
	\stackrel{\mathclap{\eqref{eq:decrease_H}}}{\leq}\\ \frac{1}{\alpha_N}H_{N,\eta}^*(x)-\frac{1}{\alpha_N}H_{N,\eta}^*(x_{\mu_{N,\eta}}(K))+\sum_{k=0}^{K-1}(c_5^\ell\left\vert r_{N,\eta}^*(x_{\mu_{N,\eta}}(k))\right\vert_{r_\mathrm{d}}^2+c_6^\ell\left\vert r_{N,\eta}^*(x_{\mu_{N,\eta}}(k))\right\vert_{r_\mathrm{d}})$ for $x_{\mu_{N,\eta}}(k)=x_{\mu_{N,\eta}}(k, x)$, and using Lemma~\ref{lem:comparison_value_functions} and Proposition~\ref{prop:performance_bound_standard_mpc} results in~\eqref{eq:transient_performance_bound} when neglecting the second term due to $\alpha_N>0$ and $H_{N,\eta}^*(x_{\mu_{N,\eta}}(K))>0$. Finally, we show that the error term $\delta$, carried over from Proposition~\ref{prop:performance_bound_standard_mpc}, is independent of $N$. In the second case of Proposition~\ref{prop:performance_bound_standard_mpc}'s proof, we had $\delta(K)=\gamma c_2^\ell c_\mathrm{s}^2\vert x\vert_{x_\mathrm{d}}^2\gamma_\mathrm{s}^{2K}$. Since here $\hat{\eta}$ and $x\in\mathcal{X}_{\tilde{N},\tilde{\eta}}$ are independent of $N$, $c_\mathrm{s}$ is fixed. Moreover, taking $\gamma_\mathrm{s}'=\sup_{N>N_\eta} \gamma_\mathrm{s}<1$ gives $\delta(K)=\gamma c_2^\ell c_\mathrm{s}^2\vert x\vert_{x_\mathrm{d}}^2(\gamma_\mathrm{s}')^{2K}\in\mathcal{L}$ independent of $N$.
\end{proof}

The first part of the transient performance bound~\eqref{eq:transient_performance_bound} containing $\alpha_N$ and $\delta(K)$ is carried over from Proposition~\ref{prop:performance_bound_standard_mpc}. The additional terms describe how close the artificial reference is chosen to the desired one, which depends mainly on $\mathbb{S}$ and the cost. Reducing these terms may not immediately result in an improved closed-loop performance, since \eqref{eq:transient_performance_bound} is an upper bound. A small prediction horizon $N$ requires an evolution of the system close to the steady state manifold. Note that a change in $\eta$ imposes a different lower bound on $N$, and quantitatively influences \eqref{eq:transient_performance_bound}, e.g. by affecting $\alpha_N$. Another significant influencing factor is $N$ itself. We illustrate these complex dependencies in Section~\ref{sec:example}.

\section{Asymptotic performance estimate}

In this section, we consider closed-loop performance in the asymptotic case, i.e.\ over an infinite time interval and if the prediction horizon approaches infinity. First, we show that if $K\to\infty$, the error terms stay bounded.

\begin{lemma}\label{lem:bounded_error_terms}
	Let Assumptions~\ref{asm:offset_cost_indication}\,--\,\ref{asm:stage_cost_difference_bound} and \ref{asm:local_exponential_cost_controllability}\,--\,\ref{asm:scaling} hold. Then, for all $\eta\geq\gamma\sigma$, $N>N_\eta$ and $x\in\mathcal{X}_{N,\eta}$ \begin{equation*}
		\sum_{k=0}^{\infty}(c_5^\ell\hspace{-1pt}\left\vert r_{N,\eta}^*(x_{\mu_{N,\eta}}(k, x))\right\vert_{r_\mathrm{d}}^2\hspace{-2pt}+c_6^\ell\hspace{-1pt}\left\vert r_{N,\eta}^*(x_{\mu_{N,\eta}}(k, x))\right\vert_{r_\mathrm{d}}\hspace{-1pt})\hspace{-2pt}<\hspace{-2pt}\infty.
	\end{equation*}
\end{lemma}
\vspace*{\belowdisplayskip}
\begin{proof}
	The proof is analogous to~\cite[Lem.~4]{Koehler2023} and follows from exponential stability of the closed loop.
\end{proof}

Next, uniform convergence of the artificial reference to the best reachable reference is shown, when considering all possible closed-loop states for a fixed set of feasible initial states. As seen in~\cite[Lem.~5]{Koehler2023}, the scaled offset cost, see Assumption~\ref{asm:scaling}, plays a crucial role in the proof.

\begin{lemma}\label{lem:uniform_convergence}
	Let Assumptions~\ref{asm:offset_cost_indication}\,--\,\ref{asm:stage_cost_difference_bound} and \ref{asm:local_exponential_cost_controllability}\,--\,\ref{asm:scaling} hold. For any $\tilde{\eta}\geq\gamma\sigma$ and $\tilde{N}>N_{\tilde{\eta}}$ there exists $\hat{\eta}$ such that $\lim_{N\to\infty}\vert r_{N,\eta}^*(k)\vert_{r_\mathrm{d}}=0$ uniformly on $\mathcal{X}_{\tilde{N},\tilde{\eta}}$ for all $k\in\mathbb{N}_0$ and $\eta\geq\hat{\eta}$, with $r_{N,\eta}^*(k)=r_{N,\eta}^*(x_{\mu_{N,\eta}}(k, x))$.
\end{lemma}
\begin{proof}
	Suppose there exists $\psi>0$ such that for all $N\in\mathbb{N}_0$ there exist $k\in\mathbb{N}_0$ and $x\in\mathcal{X}_{\tilde{N},\tilde{\eta}}$ with $\vert r_{N,\eta}^*(k)\vert_{r_\mathrm{d}}\geq\psi$. Then, $T(r_{N,\eta}^*(k))\geq\alpha_\mathrm{lo}^T(\psi)$ by Assumption~\ref{asm:offset_cost_indication}. Also, Lemma~\ref{lem:comparison_value_functions} states that there exists $\hat{\eta}$ with $H_{N,\eta}^*(x)\hspace{-1pt}\leq\hspace{-1pt}\mathcal{J}_N^\mathrm{s}(x, r_\mathrm{d})\hspace{-1pt}\leq\hspace{-1pt}\hat{\eta}$ for all $x\in\mathcal{X}_{\tilde{N},\tilde{\eta}}$, $N\in\mathbb{N}$ and $\eta\geq\hat{\eta}$. Using equation~\eqref{eq:decrease_H} recursively and nonnegative costs yield $\hat{\eta}\hspace{-1pt}\geq\hspace{-1pt} H_{N,\eta}^*(x)\hspace{-1pt}\geq\hspace{-1pt} H_{N,\eta}^*(x_{\mu_{N,\eta}}(k, x))\hspace{-1pt}=\hspace{-1pt}\mathcal{J}_N^\mathrm{s}(x_{\mu_{N,\eta}}(k, x), r_{N,\eta}^*(k))\hspace{-1pt}+\hspace{-1pt}\lambda(N)T(r_{N,\eta}^*(k))\hspace{-1pt}\geq\hspace{-1pt}\lambda(N)\alpha_\mathrm{lo}^T(\psi)$ for $N>N_\eta$. The value $\hat{\eta}$ is fixed, while $\lambda(N)\to\infty$ as $N\to\infty$ by Assumption~\ref{asm:scaling}. This contradiction proves the claim.
\end{proof}

The lemma shows that the used $\eta\geq\hat{\eta}$ has to be large enough such that $r_\mathrm{d}$ can be the artificial reference. Otherwise, the infinitely long open-loop trajectory would get so heavily restricted by the constraint of fulfilling $\mathcal{J}_N^\mathrm{s}(x, r)\leq\eta$, see~\eqref{eq:MPC_problem_costconstraint}, that it only evolves to and around a close reference $r\in\mathbb{S}$, see also Section~\ref{sec:example}. But $\hat{\eta}$ is only a minimal cost bound. Since the value function is bounded by $\hat{\eta}$, cf. Lemma~\ref{lem:comparison_value_functions}, and $N\to\infty$ fulfills the minimum horizon for stability, $\eta\geq\hat{\eta}$ are possible.

Together, the asymptotic performance result follows:

\begin{theorem}\label{thm:asymptotic_performance}
	Let Assumptions~\ref{asm:offset_cost_indication}\,--\,\ref{asm:scaling} hold. Then, for any $\tilde{\eta}\geq\gamma\sigma$ and $\tilde{N}>N_{\tilde{\eta}}$ there exists $\hat{\eta}$ such that for any $x\in\mathcal{X}_{\tilde{N},\tilde{\eta}}$ and $\eta\geq\hat{\eta}$ \begin{equation*}
		\lim_{\substack{N\to\infty\\K\to\infty}}J_K^\mathrm{d}(x, \mu_{N,\eta})=\inf_{u\in\mathbb{U}_{\lbrace x_\mathrm{d}\rbrace}^\infty(x)}J_\infty^\mathrm{d}(x, u).
	\end{equation*}
\end{theorem}
\vspace*{\belowdisplayskip}
\begin{proof}
	Let $K\to\infty$ and $N\to\infty$ in the transient performance result~\eqref{eq:transient_performance_bound}. Then, $\delta(K)$ vanishes since $\delta\in\mathcal{L}$. Combining the convergence shown in Lemma~\ref{lem:bounded_error_terms} and the uniform convergence shown in Lemma~\ref{lem:uniform_convergence} yields $\lim_{N\to\infty}\sum_{k=0}^{\infty}\left\vert r_{N,\eta}^*(x_{\mu_{N,\eta}}(k))\right\vert_{r_\mathrm{d}}^i= \sum_{k=0}^{\infty}\big(\lim_{N\to\infty}\left\vert r_{N,\eta}^*(x_{\mu_{N,\eta}}(k))\right\vert_{r_\mathrm{d}}^i\big) = 0$ for all $\eta\geq\hat{\eta}$, $x\in\mathcal{X}_{\tilde{N},\tilde{\eta}}$ and $i\in\left\lbrace 1,2\right\rbrace$. Since there exists $\gamma_\mathrm{s}'$ such that $\gamma_\mathrm{s}\leq\gamma_\mathrm{s}'<1$ for all sufficiently large and in particular all stabilizing $N$, this implies $\kappa=c_\mathrm{s}\vert x\vert_{x_\mathrm{d}}\gamma_\mathrm{s}^K\to 0$ for $K\to\infty$. Finally, $\lim_{N\to\infty}\alpha_N=1-\frac{\eta}{\sigma}(\gamma-1)/N=1$. Equality follows from optimality of the infimum, compare $u\in\mathbb{U}_{\lbrace x_\mathrm{d}\rbrace}^\infty(x)$, and stability, i.e.\ $\lim_{K\to\infty}\vert x_{\mu_{N,\eta}}(K, x)\vert_{ x_\mathrm{d}}=0$, of the MPC scheme.
\end{proof}

This theorem shows that in the asymptotic case, the closed-loop cost equals the infinite horizon optimal cost. So in the limit $N\to\infty$ and $K\to\infty$, the closed-loop trajectories of the proposed scheme are optimal among all trajectories leading to the best reachable steady state, recovering the infinite horizon optimal behaviour.

\section{Example: Continuous stirred-tank reactor}\label{sec:example}

We illustrate the effects of $\eta$ and $N$ in a simulation of a continuous stirred-tank reactor taken from~\cite{Mayne2011}. The discrete-time model results from an Euler discretization with a sampling time of \SI{0.1}{\second}. The constraints are given by $\mathbb{Z} = (\left[ 0,\,1\right]\times\left[ 0,\,1\right])\times\left[ 0,\,2\right]$ and the references have to lie in $\mathbb{Z}_r = (\left[0.0529,\,0.9492\right]\times\left[0.43,\,0.86\right])\times\left[0.1366,\,0.7687\right]$. Starting from $x_0=\begin{bmatrix}
	0.9492&0.43
\end{bmatrix}^\top$, the system should be steered to the external reference $x_\mathrm{e}=x_\mathrm{d}=\begin{bmatrix}
	0.2632&0.6519
\end{bmatrix}^\top$ with $u_\mathrm{e}=u_\mathrm{d}=0.7585$, which in this case is equal to the best reachable reference. We choose a quadratic stage cost with $Q=I$, $R=0$, the offset cost $T(r)=0.01\vert x_{r,1}\vert_{x_{\mathrm{d},1}}^2+1000\vert x_{r,2}\vert_{x_{\mathrm{d},2}}^2+\vert u_r\vert_{u_\mathrm{d}}^2$ and the scaling function $\lambda(N)=N+1$. The simulation was implemented using~\cite{Andersson2019,Waechter2006}. We did not check Assumption~\ref{asm:local_exponential_cost_controllability} or the requirements of Theorems~\ref{thm:exponential_stability}--\ref{thm:asymptotic_performance}, but, as in \cite{Soloperto2023}, chose $N$ and $\eta$ by trial and error. Modifying $Q$ and $R$ affects both the required values of $\eta$ and $N$ for stability, as well as the specified performance measure~\eqref{eq:performance_measure}.

For $\eta=0.5$, see the top of Figure~\ref{fig:CSTR}, the system shows stable behaviour for small horizons $N$ in contrast to standard MPC without terminal constraints ($N\geq 30$), which also benefits computational complexity. Even though the performance gets closer to the infinite horizon performance for increasing $N$, it decreases again when increasing $N$ further, see the top of Figure~\ref{fig:CSTR}. A too small $\eta$ for a large $N$ forces the open-loop solution towards an artificial steady state close to the current state in order to reduce costs, cf. Lemma~\ref{lem:uniform_convergence}, which leads to suboptimal behaviour. Thus, \eqref{eq:MPC_problem_costconstraint} can play a similar role as stabilizing terminal constraints in certain situations. In contrast, when choosing a larger $\eta=10$\footnote{Note that this is several times smaller than the infinite horizon optimal closed-loop cost $J_\infty^\mathrm{d}(x_0, u_\infty^\mathrm{s}(0\vert t))\approx 208.9749$.}, for rising $N$, the infinite horizon optimal solution is attained, see the bottom of Figure~\ref{fig:CSTR}, as expected from Theorem~\ref{thm:asymptotic_performance}. Also, MPC for tracking for $\lambda(N)\equiv 1$ and $N=1000$ is shown for comparison.

\begin{figure}
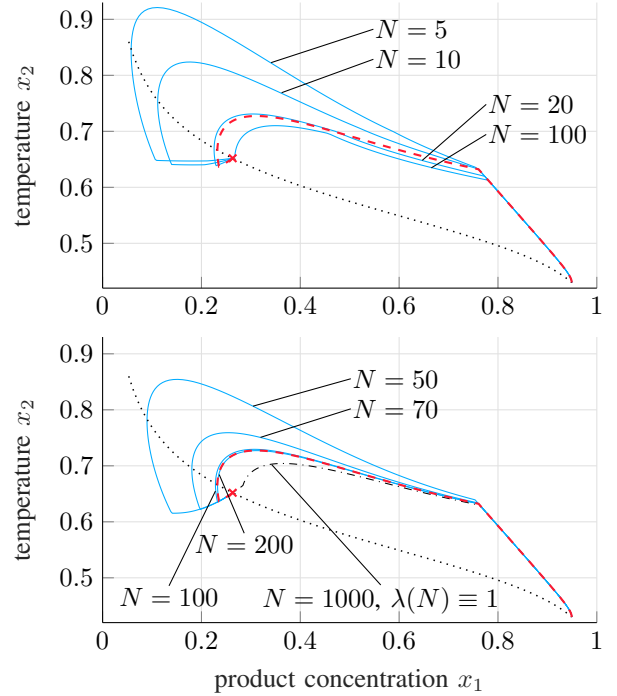

	\centering
	\vspace{5pt}
	\newlength{\figurewidth}
	\setlength\figurewidth{0.8\columnwidth}
	\input{CSTR_eta05.tex}\\
	\vspace{5pt}
	\input{CSTR_eta10.tex}
	\caption{Closed-loop trajectories (blue) from application of MPC for tracking for $\eta=0.5$ (top), $\eta=10$ (bottom) and varying prediction horizons $N$. The dotted black line indicates possible artificial steady states, while the infinite horizon optimal solution is dashed in red.}
	\label{fig:CSTR}
	\vspace{-12pt}
\end{figure}

\section{Conclusion}

We developed a transient performance estimate for an MPC for tracking scheme without terminal cost and constraints. Our analysis showed that for good performance, $\eta$ needs to be sufficiently large for $r_\mathrm{d}$ to be used as the artificial reference. The system then has enough flexibility to perform well, as seen in the simulation. For sufficiently large $\eta$, larger $N$ lead to better performance. However, $\eta$ should not be chosen excessively large, as this may conflict with practical horizon limitations. It would be interesting to find a relation between them in order to design only one parameter.

\appendix
\subsection{Proof of Proposition~\ref{prop:performance_bound_standard_mpc}}\label{appendix:proof_standard_performance}

Use $K_\kappa\hspace{-2pt}=\hspace{-2pt}\min\lbrace K\hspace{-2pt}\in\hspace{-2pt}\mathbb{N}\hspace{1pt}\vert\hspace{2pt}\kappa\hspace{-2pt}=\hspace{-1pt}c_\mathrm{s}\vert x\vert_{x_\mathrm{d}}\hspace{-1pt}\gamma_\mathrm{s}^K\hspace{-2pt}\leq\hspace{-4pt}\sqrt{\frac{\sigma}{c_2^\ell}}\rbrace$. For $K<K_\kappa$ and $K\leq N$, \eqref{eq:performance_bound_standard_mpc} can always be achieved by choosing $\delta(K)$ sufficiently large, since all functions~are bounded as $0\leq \mathcal{J}_N^\mathrm{s}(x, r_\mathrm{d})\leq \eta$ and $0\leq J_K^\mathrm{d}(x, u)$ for all $u$.
	
In the case that $K\geq K_\kappa$ but still $K\leq N$, consider some $\bar{u}\in\mathbb{U}_{\mathcal{B}_\kappa(x_\mathrm{d})}^K(x)$. From \eqref{eq:stage_cost_lower_and_upper_bound} follows $\ell^*(x_{\bar{u}}(K, x), r_\mathrm{d})\leq c_2^\ell\vert x_{\bar{u}}(K, x)\vert_{x_\mathrm{d}}^2 \leq c_2^\ell \kappa^2\leq c_2^\ell \frac{\sigma}{c_2^\ell}=\sigma$, hence Assumption~\ref{asm:local_exponential_cost_controllability} yields $\mathcal{J}_{N-K}^\mathrm{s}(x_{\bar{u}}(K, x), r_\mathrm{d})\hspace{-1pt}\leq\hspace{-1pt}\gamma\ell^*(x_{\bar{u}}(K, x), r_\mathrm{d})\hspace{-1pt}\leq\hspace{-1pt}\gamma c_2^\ell\vert x_{\bar{u}}(K, x)\vert_{x_\mathrm{d}}^2\leq\gamma c_2^\ell\kappa^2=\gamma c_2^\ell c_\mathrm{s}^2\vert x\vert_{x_\mathrm{d}}^2\gamma_\mathrm{s}^{2K}$. Thus, $x_{\bar{u}}(K, x)\in\mathcal{X}_{N-K}^\mathrm{s}$ and $\bar{u}\in\mathbb{U}_{\mathcal{X}_{N-K}^\mathrm{s}}^K(x)$. Utilizing the dynamic programming principle~\cite[Thm.~3.15]{Gruene2017} gives for $x\in\mathcal{X}_N^\mathrm{s}$, which is satisfied because of $x$ with $\mathcal{J}_N^\mathrm{s}(x, r_\mathrm{d})\leq\eta$ since the value function maps infeasible $x$ to infinity, for all $N\in\mathbb{N}$ and $K\in\mathbb{I}_{1:N}$ with $x_u(k)=x_u(k, x)$,\\
$\mathcal{J}_N^\mathrm{s}(x, r_\mathrm{d}) =\hspace{-2pt} \inf_{u(\cdot)\in\mathbb{U}_{\mathcal{X}_{N-K}^\mathrm{s}}^K\hspace{-2pt}(x)}\hspace{-1pt}\big\lbrace \sum_{k=0}^{K-1}\ell(x_u(k), u(k), r_\mathrm{d})\\
+\mathcal{J}_{N-K}^\mathrm{s}(x_u(K), r_\mathrm{d}) \big\rbrace
\leq J_K(x, \bar{u}, r_\mathrm{d})+\mathcal{J}_{N-K}^\mathrm{s}(x_{\bar{u}}(K), r_\mathrm{d})$. Take $\bar{u}=\arginf_{u\in\mathbb{U}_{\mathcal{B}_\kappa(x_\mathrm{d})}^K(x)}J_K(x, u, r_\mathrm{d})$ since $\bar{u}$ was arbitrary. This results in \eqref{eq:performance_bound_standard_mpc} with $\delta(K)=\gamma c_2^\ell c_\mathrm{s}^2\vert x\vert_{x_\mathrm{d}}^2\gamma_\mathrm{s}^{2K}$. Since $\gamma_\mathrm{s}\in(0, 1)$ and $\delta(K)$ is continuous, $\delta(K)\in\mathcal{L}$.
	
In the case that $K>N$ and either $K<K_\kappa$ or $K\geq K_\kappa$, we have $J_K^\mathrm{d}(x, u)=J_K(x, u, r_\mathrm{d})\geq \mathcal{J}_N^\mathrm{s}(x, r_\mathrm{d})$ for all $u\in\mathbb{U}_{\mathcal{B}_\kappa(x_\mathrm{d})}^K(x)$, since $\ell(x, u, r_\mathrm{d})\geq 0$ for $(x, u)\in\mathbb{Z}$, so in total the claimed inequality is satisfied for all $K\in\mathbb{N}$.

\subsection{Proof of Lemma~\ref{lem:comparison_value_functions}}\label{appendix:proof_comparison_value_functions}

Case 1: $\mathcal{J}_{\tilde{N}}^\mathrm{s}(x, r_\mathrm{d})=J_{\tilde{N}}(x, u_{\tilde{N}}^\mathrm{s}, r_\mathrm{d})\leq\tilde{\eta}$ for all $x\in\mathcal{X}_{\tilde{N},\tilde{\eta}}$. Choose $r_\mathrm{d}$ and $u_{\tilde{N}}^\mathrm{s}$ as a candidate for~\eqref{eq:MPC_problem}. From~\eqref{eq:ell_star_lower_sigma} follows that there exists $N'\in\mathbb{I}_{1:\tilde{N}-1}$ such that $\ell^*(x_\mathrm{end}, r_\mathrm{d})<\sigma$ with $x_\mathrm{end}=x_{u_{\tilde{N}}^\mathrm{s}}(N', x)$. Name $\eta'=\tilde{\eta}$ and $\hat{u}_1=\big[u_{\tilde{N}}^\mathrm{s}(0\vert 0), \dots, u_{\tilde{N}}^\mathrm{s}(N'-1\vert 0)\big]$.
	
Case 2: $\mathcal{J}_{\tilde{N}}^\mathrm{s}(x, r_\mathrm{d})>\tilde{\eta}$ for $x\in\mathcal{X}_{\tilde{N},\tilde{\eta}}$. Then, there is $r_1\neq r_\mathrm{d}$ with $H_{\tilde{N},\tilde{\eta}}^*(x)=\mathcal{J}_{\tilde{N}}^\mathrm{s}(x, r_1)+\lambda(\tilde{N})T(r_1)$ and $\mathcal{J}_{\tilde{N}}^\mathrm{s}(x, r_1)\leq\tilde{\eta}$. From~\eqref{eq:ell_star_lower_sigma} follows that there is $\tau_1\in\mathbb{I}_{1:\tilde{N}-1}$ with $\ell^*(x_{\tau_1}, r_1)<\sigma$ for $x_{\tau_1}=x_{u_{\tilde{N}}^\mathrm{s}}(\tau_1, x)$. Next, uniform convergence to the neighbourhood around $r_\mathrm{d}$ is shown, which is similar to \cite[Lem.~3]{Koehler2023}, but for the case without terminal ingredients. In this part, we use $\nu_{N,r}(x(t))=u_N^\mathrm{s}(0\vert t)$, where $u_N^\mathrm{s}$ is the solution of~\eqref{eq:standard_mpc_problem} given $x$, $r$ and $N$. When considering the standard MPC problem~\eqref{eq:standard_mpc_problem} with a fixed $\tilde{N}$ and general $r$, an initial state $x$ in the neighbourhood $\ell^*(x, r)\leq\sigma$ implies $\vert x\vert_{x_r}^2\leq\frac{\sigma}{c_1^\ell}$ by Assumption~\ref{asm:stage_cost_lower_and_upper_bound}. Also, recall that $c_\mathrm{s}$ and $\gamma_\mathrm{s}$ in \eqref{eq:exponential_stability_standard_mpc} can be chosen independently of the reference. Thus, there is $\tau_2$ such that $\vert x_{\nu_{\tilde{N},r}}(\tau_2, x)\vert_{x_r}\leq c_\mathrm{s}\sqrt{\frac{\sigma}{c_1^\ell}}\gamma_\mathrm{s}^{\tau_2}\leq\sqrt{\frac{\sigma}{4c_2^\ell}}$, i.e.\ after following the standard MPC feedback law for $\tau_2$ steps the state is in a region even closer around $r$. Utilize this for $r_1$ first and take $\delta_r=\sup_{\hat{r},\tilde{r}\in\mathbb{S}}\vert\hat{r}\vert_{\tilde{r}}$. By Assumption~\ref{asm:better_candidate_reference} there exists $r_2$ with $\vert x_{r_1}\vert_{x_{r_2}}^2\leq \vert r_1\vert_{r_2}^2\leq(c_1^r)^2\theta^2\vert r_2\vert_{r_\mathrm{d}}^2\leq (c_1^r)^2\theta^2\delta_r^2\leq\frac{\sigma}{4c_2^\ell}$ when choosing $0<\theta\leq\sqrt{\frac{\sigma}{4c_2^\ell}}\frac{1}{c_1^r\delta_r}<1$. With \eqref{eq:stage_cost_lower_and_upper_bound} follows $\ell^*(x_{\nu_{\tilde{N},r_1}}(\tau_2, x_{\tau_1}), r_2)\leq c_2^\ell\vert x_{\nu_{\tilde{N},r_1}}(\tau_2, x_{\tau_1})\vert_{x_{r_2}}^2\leq 2c_2^\ell(\vert x_{\nu_{\tilde{N},r_1}}(\tau_2, x_{\tau_1})\vert_{x_{r_1}}^2+\vert x_{r_1}\vert_{x_{r_2}}^2)\leq\sigma$, and Assumption~\ref{asm:local_exponential_cost_controllability} yields $\mathcal{J}_{\tilde{N}}^\mathrm{s}(x_{\nu_{\tilde{N},r_1}}(\tau_2, x_{\tau_1}), r_2)\leq\gamma\ell^*(x_{\nu_{\tilde{N},r_1}}(\tau_2, x_{\tau_1}), r_2)\leq\gamma\sigma\leq \tilde{\eta}$. Hence $x_{\nu_{\tilde{N},r_1}}(\tau_2, x_{\tau_1})\in\mathcal{X}_{\tilde{N},\tilde{\eta}}$ and $\in\mathcal{X}_{\tilde{N}}^\mathrm{s}$. Take $j=1$. If $\vert r_j\vert_{r_\mathrm{d}}^2\leq\frac{\sigma}{4c_2^\ell}$, we are finished with $r_\mathrm{d}=r_{j+1}$ and $j+1=c_\tau$. Otherwise, if $\vert r_j\vert_{r_\mathrm{d}}^2>\frac{\sigma}{4c_2^\ell}$, follow $\nu_{\tilde{N},r_{j+1}}$ for $\tau_2$ steps and repeat these steps until $\vert r_j\vert_{r_\mathrm{d}}^2\leq\frac{\sigma}{4c_2^\ell}$ is satisfied. Then, we can concatenate the inputs $\nu_{\tilde{N},r_j}$ for $j\in\mathbb{I}_{1:c_\tau-1}$ to get from the initial neighbourhood into the neighbourhood around $r_\mathrm{d}$. The minimum required number of artificial references $c_\tau$ results from applying Assumption~\ref{asm:better_candidate_reference} recursively $T(r_j)<T(r_1)-(j-1)c_2^r\theta\frac{\sigma}{4c_2^\ell}$ and evaluating it for $c_\tau-2$ together with Assumption~\ref{asm:offset_cost_indication}: $\alpha_\mathrm{lo}^T\big(\sqrt{\frac{\sigma}{4c_2^\ell}}\,\big)<\alpha_\mathrm{lo}^T(\vert r_{c_\tau-2}\vert_{r_\mathrm{d}})\leq T(r_{c_\tau-2})<T(r_1)-(c_\tau-3)c_2^r\theta\frac{\sigma}{4c_2^\ell}\leq\alpha_\mathrm{up}^T(\vert r_1\vert_{r_\mathrm{d}})-(c_\tau-3)c_2^r\theta\frac{\sigma}{4c_2^\ell}\leq\alpha_\mathrm{up}^T(\delta_r)-(c_\tau-3)c_2^r\theta\frac{\sigma}{4c_2^\ell}$. It is given by $c_\tau \geq 3+\frac{\alpha_\mathrm{up}^T(\delta_r)-\alpha_\mathrm{lo}^T(\sqrt{\frac{\sigma}{4c_2^\ell}})}{c_2^r\theta\frac{\sigma}{4c_2^\ell}}$. In total, the candidate $\hat{u}_1$ consists of parts of $u_{\tilde{N}}^\mathrm{s}$ and $\nu_{\tilde{N},r_1}, \dots, \nu_{\tilde{N},r_{c_\tau-1}}$, has length $N'=\tau_1+(c_\tau-1)\tau_2$ and transfers any $x\in\mathcal{X}_{\tilde{N},\tilde{\eta}}$ into the neighbourhood $\ell^*(x_\mathrm{end}, r_\mathrm{d})\leq\sigma$ with $x_\mathrm{end}=x_{\hat{u}_1}(N', x)$. Moreover, the associated cost is smaller than or equal to $\eta'=\max_{x\in\mathcal{X}_{\tilde{N},\tilde{\eta}}\subseteq X}J_{N'}(x, \hat{u}_1, r_\mathrm{d})$. This value is an overestimation because the candidate trajectory moves through the neighbourhoods along the steady state manifold, which may not be the optimal behaviour.
	
In both cases, Assumption~\ref{asm:local_exponential_cost_controllability} implies that there is an arbitrarily long $\hat{u}_2=u_{\hat{N}}^\mathrm{s}(\cdot\vert N')$ of \eqref{eq:standard_mpc_problem} with $r=r_\mathrm{d}$ and starting from $x=x_\mathrm{end}$ without increasing the cost beyond a fixed constant $\mathcal{J}_{\hat{N}}^\mathrm{s}(x_\mathrm{end}, r_\mathrm{d})=\mathcal{J}_{\hat{N}}^\mathrm{s}(x_\mathrm{end}, \hat{u}_2, r_\mathrm{d}) \leq\gamma\ell^*(x_\mathrm{end}, r_\mathrm{d})\leq\gamma\sigma$. Combined, there is $\hat{\eta}$ such that $\mathcal{J}_N^\mathrm{s}(x, r_\mathrm{d})\leq\hat{\eta}$ for all $N\in\mathbb{N}$ and $x\in\mathcal{X}_{\tilde{N},\tilde{\eta}}$, when choosing $\hat{\eta}=\eta'+\gamma\sigma$. For $N>N'$ this holds, because with the constructed candidate $\hat{u}=\left[\hat{u}_1,\hat{u}_2\right]$ the value function satisfies $\mathcal{J}_N^\mathrm{s}(x, r_\mathrm{d})\leq J_N(x, \hat{u}, r_\mathrm{d})=J_{N'}(x, \hat{u}_1, r_\mathrm{d})+J_{N-N'}(x_\mathrm{end}, \hat{u}_2, r_\mathrm{d})\leq \eta'+\gamma\sigma=\hat{\eta}$. For $N\leq N'$, the inequality $\mathcal{J}_N^\mathrm{s}(x, r_\mathrm{d})\leq \mathcal{J}_{N'}^\mathrm{s}(x, r_\mathrm{d})\leq J_{N'}(x, \hat{u}_1, r_\mathrm{d})\leq\eta'<\hat{\eta}$ holds due to nonnegative stage costs and the candidate $\hat{u}_1$. This fulfills constraint~\eqref{eq:MPC_problem_costconstraint} for $\eta\geq\hat{\eta}$. Thus, the solution of the standard MPC problem~\eqref{eq:standard_mpc_problem} w.r.t. $r_\mathrm{d}$ is a candidate in~\eqref{eq:MPC_problem} with modified parameters $N\in\mathbb{N}$ and $\eta\geq\hat{\eta}$, for $x\in\mathcal{X}_{\tilde{N},\tilde{\eta}}$, and with $T(r_\mathrm{d})=0$, i.e.\ $H_{N,\eta}^*(x)\leq \mathcal{J}_N^\mathrm{s}(x, r_\mathrm{d})+\lambda(N)T(r_\mathrm{d}) = \mathcal{J}_N^\mathrm{s}(x, r_\mathrm{d})\leq \hat{\eta}$. Moreover, \eqref{eq:MPC_problem} and \eqref{eq:standard_mpc_problem} remain feasible for the modified parameters, hence $\mathcal{X}_{\tilde{N},\tilde{\eta}}\subseteq\mathcal{X}_{N,\eta}$ and $\mathcal{X}_{\tilde{N},\tilde{\eta}}\subseteq\mathcal{X}_N^\mathrm{s}$.

\bibliographystyle{IEEEtran}
\bibliography{literature}

\begin{thebibliography}{10}
\providecommand{\url}[1]{#1}
\csname url@samestyle\endcsname
\providecommand{\newblock}{\relax}
\providecommand{\bibinfo}[2]{#2}
\providecommand{\BIBentrySTDinterwordspacing}{\spaceskip=0pt\relax}
\providecommand{\BIBentryALTinterwordstretchfactor}{4}
\providecommand{\BIBentryALTinterwordspacing}{\spaceskip=\fontdimen2\font plus
\BIBentryALTinterwordstretchfactor\fontdimen3\font minus
  \fontdimen4\font\relax}
\providecommand{\BIBforeignlanguage}[2]{{%
\expandafter\ifx\csname l@#1\endcsname\relax
\typeout{** WARNING: IEEEtran.bst: No hyphenation pattern has been}%
\typeout{** loaded for the language `#1'. Using the pattern for}%
\typeout{** the default language instead.}%
\else
\language=\csname l@#1\endcsname
\fi
#2}}
\providecommand{\BIBdecl}{\relax}
\BIBdecl

\bibitem{Gruene2017}
L.~Grüne and J.~Pannek, \emph{{Nonlinear Model Predictive Control: Theory and
  Algorithms}}, 2nd~ed., ser. Commun. Control Eng.\hskip 1em plus 0.5em minus
  0.4em\relax Cham, Switzerland: Springer, 2017.

\bibitem{Limon2008}
D.~Limon, I.~Alvarado, T.~Alamo, and E.~F. Camacho, ``{MPC for tracking
  piecewise constant references for constrained linear systems},''
  \emph{Automatica}, vol.~44, no.~9, pp. 2382--2387, 2008.

\bibitem{Limon2018}
D.~Limon, A.~Ferramosca, I.~Alvarado, and T.~Alamo, ``{Nonlinear MPC for
  Tracking Piece-Wise Constant Reference Signals},'' \emph{IEEE Trans. Autom.
  Control}, vol.~63, no.~11, pp. 3735--3750, 2018.

\bibitem{Koehler2020b}
J.~Köhler, M.~A. Müller, and F.~Allgöwer, ``{A nonlinear tracking model
  predictive control scheme for dynamic target signals},'' \emph{Automatica},
  vol. 118, p. 109030, 2020.

\bibitem{Krupa2024}
{P. Krupa et al.}, ``{Model predictive control for tracking using artificial
  references: Fundamentals, recent results and practical implementation},'' in
  \emph{Proc. IEEE 63rd Conf. Decis. Control (CDC)}, Milan, Italy, 2024, pp.
  2977--2991.

\bibitem{Koehler2024}
J.~Köhler, M.~A. Müller, and F.~Allgöwer, ``{Analysis and design of model
  predictive control frameworks for dynamic operation-An overview},''
  \emph{Annu. Rev. Control}, vol.~57, p. 100929, 2024.

\bibitem{Koehler2020a}
J.~K\"ohler, M.~A. M\"uller, and F.~Allg\"ower, ``{A Nonlinear Model Predictive
  Control Framework Using Reference Generic Terminal Ingredients},'' \emph{IEEE
  Trans. Autom. Control}, vol.~65, no.~8, pp. 3576--3583, 2020.

\bibitem{Soloperto2023}
R.~Soloperto, J.~Köhler, and F.~Allgöwer, ``{A Nonlinear MPC Scheme for
  Output Tracking Without Terminal Ingredients},'' \emph{IEEE Trans. Autom.
  Control}, vol.~68, no.~4, pp. 2368--2375, 2023.

\bibitem{Ferramosca2009}
A.~Ferramosca, D.~Limon, I.~Alvarado, T.~Alamo, and E.~F. Camacho, ``{MPC for
  tracking with optimal closed-loop performance},'' \emph{Automatica}, vol.~45,
  no.~8, pp. 1975--1978, 2009.

\bibitem{Ferramosca2011}
A.~Ferramosca, D.~Limon, I.~Alvarado, T.~Alamo, F.~Castaño, and E.~F. Camacho,
  ``{Optimal MPC for tracking of constrained linear systems},'' \emph{Int. J.
  Syst. Sci.}, vol.~42, no.~8, pp. 1265--1276, 2011.

\bibitem{Koehler2023}
M.~Köhler, L.~Krügel, L.~Grüne, M.~A. Müller, and F.~Allgöwer,
  ``{Transient Performance of MPC for Tracking},'' \emph{IEEE Control Syst.
  Lett.}, vol.~7, pp. 2545--2550, 2023.

\bibitem{Kellett2014}
C.~M. Kellett, ``A compendium of comparison function results,'' \emph{Math.
  Control Signals Syst.}, vol.~26, no.~3, pp. 339--374, 2014.

\bibitem{Gruene2012}
L.~Grüne, ``{NMPC without terminal constraints},'' \emph{IFAC Proc. Volumes},
  vol.~45, no.~17, pp. 1--13, 2012, 4th IFAC Conf. Nonlinear Model Predictive
  Control.

\bibitem{Boccia2014}
A.~Boccia, L.~Grüne, and K.~Worthmann, ``{Stability and feasibility of state
  constrained MPC without stabilizing terminal constraints},'' \emph{Syst.
  Control Lett.}, vol.~72, pp. 14--21, 2014.

\bibitem{Gruene2008}
L.~Gr{\"u}ne and A.~Rantzer, ``On the infinite horizon performance of receding
  horizon controllers,'' \emph{IEEE Trans. Autom. Control}, vol.~53, no.~9, pp.
  2100--2111, 2008.

\bibitem{Worthmann2016}
K.~Worthmann, M.~W. Mehrez, M.~Zanon, G.~K.~I. Mann, R.~G. Gosine, and
  M.~Diehl, ``{Model Predictive Control of Nonholonomic Mobile Robots Without
  Stabilizing Constraints and Costs},'' \emph{IEEE Trans. Control Syst.
  Technol.}, vol.~24, no.~4, pp. 1394--1406, 2016.

\bibitem{Mayne2011}
D.~Q. Mayne, E.~C. Kerrigan, E.~J. van Wyk, and P.~Falugi, ``Tube-based robust
  nonlinear model predictive control,'' \emph{Int. J. Robust Nonlinear
  Control}, vol.~21, no.~11, pp. 1341--1353, 2011.

\bibitem{Andersson2019}
J.~A.~E. Andersson, J.~Gillis, G.~Horn, J.~B. Rawlings, and M.~Diehl,
  ``{CasADi} -- {A} software framework for nonlinear optimization and optimal
  control,'' \emph{Math. Program. Comput.}, vol.~11, no.~1, pp. 1--36, 2019.

\bibitem{Waechter2006}
A.~Wächter and L.~T. Biegler, ``{On the implementation of an interior-point
  filter line-search algorithm for large-scale nonlinear programming},''
  \emph{Math. Program.}, vol. 106, no.~1, pp. 25--57, 2006.

\end{thebibliography}

\end{document}